\documentclass[aps,pra,notitlepage,twocolumn,longbibliography]{revtex4-2}

\usepackage{amsmath,amssymb,amsthm,bm}
\usepackage{graphicx}
\usepackage[colorinlistoftodos]{todonotes}
\usepackage[inline]{enumitem}
\definecolor{darkblue}{rgb}{0.1,0.2,0.6}
\definecolor{darkred}{rgb}{0.8,0.1,0.2}
\definecolor{crimson}{RGB}{164,16,52}
\definecolor{darkgreen}{rgb}{0.31,0.62,0.24}
\usepackage[colorlinks=true, allcolors=crimson]{hyperref}
\usepackage{diagbox}
\usepackage{epigraph}
\usepackage{float}
\usepackage{multirow}
\usepackage{tikz}
\usepackage{mathtools}
\usepackage{braket}
\usepackage{soul} 
\usepackage[all]{xy}
\usepackage{lipsum}

\usepackage{bbm}

\newcommand{\Rom}[1]{\uppercase\expandafter{\romannumeral#1}}

\newcommand{\pa}{\partial}

\newcommand{\mc}{\mathcal}

\DeclareMathOperator{\Tr}{Tr}

\makeatletter
\renewcommand*\env@matrix[1][*\c@MaxMatrixCols c]{%
	\hskip -\arraycolsep
	\let\@ifnextchar\new@ifnextchar
	\array{#1}}
\makeatother

\newtheorem{lemma}{Lemma}

\newtheorem*{claim*}{Claim}

\newcommand{\dia}[3]{\raisebox{#2pt}{\includegraphics{dia_#1}}\hspace{#3pt}}

\usepackage[normalem]{ulem}

\definecolor{darkred}{rgb}{0.8,0.1,0.2}

\begin{document}
\title{Anyon Quantum Dimensions from an Arbitrary Ground State Wave Function}
\author{Shang Liu}
\email{sliu.phys@gmail.com}
\affiliation{Kavli Institute for Theoretical Physics, University of California, Santa Barbara, California 93106, USA}

\begin{abstract}
	Realizing topological orders and topological quantum computation is a central task of modern physics. An important but notoriously hard question in this endeavor is how to diagnose topological orders that lack conventional order parameters. 
	A breakthrough in this problem is the discovery of topological entanglement entropy, which can be used to detect nontrivial topological order from a ground state wave function, but is far from enough for fully determining the topological order. In this work, we take a key step further in this direction: We propose a simple entanglement-based protocol for extracting the quantum dimensions of all anyons from a single ground state wave function in two dimensions. The choice of the space manifold and the ground state is arbitrary. 
	This protocol is both validated in the continuum and verified on lattices, and we anticipate it to be realizable in various quantum simulation platforms. 
\end{abstract}

\maketitle
\section*{Introduction}
Topologically ordered phases of matter exhibit a number of remarkable properties, such as the existence of fractionalized excitations dubbed anyons, and robust ground state degeneracies on topologically nontrivial spaces \cite{WenBook}. From a practical perspective, they are also promising platforms for fault-tolerant quantum computation \cite{Kitaev2003ToricCode,Nayak2008RMP}. 

Unlike symmetry breaking orders, topological orders (TOs) lack conventional order parameters. They do not even require any symmetry and sometimes support gapped boundaries. Therefore, diagnosing TOs is generically a difficult task. 
Recent advances in quantum simulating topologically ordered states  \cite{Rhine2021KagomeRydberg,Ruben2021RubyRydberg,TC_SC_2021,TC_Rydberg_2021,Nat2021KW,Nat2021EfficientTO,Nat2022Shortest,Dreyer2023TC,Dreyer2023D4} have further highlighted the need for efficient protocols to identify them. 
A breakthrough in this problem is the discovery of topological entanglement entropy (EE) \cite{Hamma2005TCEE1,Hamma2005TCEE2,Kitaev2006TEE,Levin2006TEE}. It is shown that the EE of a disk region in a two-dimensional gapped ground state wave function contains a universal term dubbed the topological EE, from which we can read off the so-called total quantum dimension $\mc D$ of the system. $\mc D=1$ ($\mc D>1$) for a trivial (nontrivial) TO, and hence the topological EE can be used for detecting nontrivial TOs. However, $\mc D$ is still far from fully characterizing a TO. In particular, it can not distinguish abelian and nonabelian TOs which have very different properties and applications. 
There have been efforts to extract other universal quantities of a TO, such as the chiral central charge either from edge thermal transport \cite{Kane1997ThermalTransport,Cappelli2002ThermalConductance,Kitaev2006Honeycomb} or the bulk wave function \cite{Kim2022ChiralCentralCharge,Kim2022ChiralCCNumerics,Zaletel2022MarkovGap,Fan2022ChiralCentralCharge}, the higher central charge \cite{Ryu2024HigherCC}, and the many-body Chern number \cite{Hafezi2021ChernNumber,Barkeshli2021ChernNumber,Fan2023HallConductance}. However, these quantities again do not distinguish abelian and nonabelian TOs, and vanish for many TOs supporting gapped boundaries. 

If we know the quantum dimensions $d_j$ of all anyon types $j$, we will be able to tell apart abelian and nonabelian TOs, because the former have $d_j=1$ for all $j$, while the latter have $d_j>1$ for some $j$. Intuitively, $d_j$ is the Hilbert space dimension shared by each type-$j$ anyon in the limit of many anyons. More precisely, let $D_j(M)$ be the degeneracy of a particular anyon configuration with $M$ type-$j$ anyons. Then in the large $M$ limit, $D_j(M)/d_j^M$ is of order $1$ \cite{Nayak2008RMP}. $\mc D$ is related to $d_j$'s by $\mc D^2=\sum_j d_j^2$. $d_j$'s impose nontrivial constraints on the fusion rules of anyons, and if the chiral central charge is known, they also constrain the anyon self-statistics \cite{Kitaev2006Honeycomb}. Therefore, $d_j$'s are important quantitative characterizations of the anyon excitations. 

In this paper, we propose a very simple protocol for extracting the quantum dimensions of all anyons from an arbitrary ground state of a TO on an arbitrary space, e.g. a disk. 
There are other existing methods to extract $d_j$ \cite{Kitaev2006Honeycomb,Kitaev2006TEE,Fradkin2008TEE,Vishwanath2012TEE,MeiWen2015ModularMats,Wen2016NegBdry,Wen2016NegSurgery,Kim2020EntBootstrapFusion,YinLiu2023Ent} as well. Some of these methods require knowing the operators for creating anyons, which is unlikely in the case of an unknown wave function. Some other methods require particular state(s) on a torus, which is harder to experimentally prepare than states on a planar geometry. The approach of Ref.\,\onlinecite{Kim2020EntBootstrapFusion} does not need either of these two, but requires accessing some infinite set of density matrices, which is more of conceptual than practical significance. The key outstanding feature of our proposal is that we have avoided the aforementioned requirements. 

In the rest of the paper, we will first describe our protocol, then justify it with a field-theoretic approach, and finally test it on lattices using Kitaev's quantum double models \cite{Kitaev2003ToricCode}. 

\section*{Results}
\subsection*{Protocol}
Consider a two-dimensional topologically ordered system on an arbitrary space manifold with or without a boundary. Let $\ket{\psi}$ be any state with no excitations in a large enough region, say a ground state. We will describe and later justify an efficient protocol for extracting the quantum dimensions of all anyons. 

\begin{figure}
	\centering
	\includegraphics{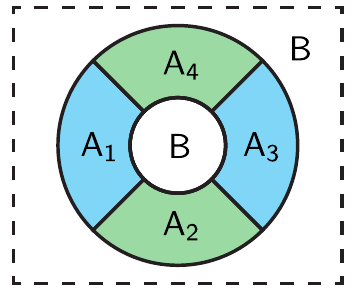}
	\caption{The partition of space used in our protocol. There is an annulus region $A=\bigcup_{i=1}^4 A_i$ subdivided into four parts. The remaining uncolored region is $B$. }
	\label{fig:AnnulusRegion}
\end{figure}
Consider a partition of the space as shown in Fig.\,\ref{fig:AnnulusRegion} in a region with no excitations. $A=\bigcup_{i=1}^4 A_i$ takes an annulus shape, and $B$ is the rest of the system. Our protocol consists of three steps listed below. Note that we will first describe the protocol as if we are performing an analytical or numerical computation. A possible experimental realization will be given later. 
\begin{itemize}
	\item \textbf{Step 1:} Obtain the reduced density matrix $\rho_A:=\Tr_B\ket{\psi}\bra{\psi}$ for the annulus region $A$. 
	\item \textbf{Step 2:} Map $\rho_A$ to a pure state in the doubled Hilbert space: Let $\rho_A=\sum_{i,j}M_{ij}\ket{i}\bra{j}$ where $\{ \ket{i} \}$ is an arbitrary real-space tensor product basis for Region $A$. We define
	\begin{align}
		\ket{\rho_A}&:=\frac{1}{\sqrt{\Tr(\rho_A^2)}}\sum_{i,j}M_{ij}\ket{i}\ket{j}. 
	\end{align}
	
	\item \textbf{Step 3:} Denote the doubled system by $A\cup A'$, and divide $A'$ as $\bigcup_{i=1}^4 A_i'$ in the same way as $A$. Compute the Renyi mutual information $I^{(n)}$ (defined later) between $A_1\cup A_1'$ and $A_3\cup A_3'$ for several different Renyi indices $n$, and solve the anyon quantum dimensions $d_j$ according to the following formula. 
	\begin{align}
		I^{(n)}(11',33')=\frac{1}{n-1}\log\left[ \sum_j\left(\frac{d_j}{\mc D}\right)^{4-2n} \right], 
		\label{eq:ProtocolFormula}
	\end{align}
	where $ii'$ stands for $A_i\cup A_i'$, and $\mc D:=\sqrt{\sum_j d_j^2}$ is the total quantum dimension. 
\end{itemize}
Here, the Renyi mutual information is defined as usual by  $I^{(n)}(X,Y):=S^{(n)}_X+S^{(n)}_Y-S^{(n)}_{X\cup Y}$, where $S^{(n)}_P:=(1-n)^{-1}\log\Tr(\rho_P^n)$ is the Renyi entropy. 

Intuitively, $\ket{\rho_A}$ is a particular ground state of the TO on the torus obtained by gluing $A$ and $A'$ along their boundaries. We will later carefully justify this picture and determine this special state. Once this is done, Eq.\,\ref{eq:ProtocolFormula} follows from a known result about mutual information on torus \cite{Wen2016NegSurgery}. 

A few comments are in order. In Step 2, the map from $\rho_A$ to $\ket{\rho_A}$ is basis dependent. If we choose a different real-space tensor product basis, then the new pure state $\ket{\rho_A}_{\rm new}$ is related to the old one by a {local} basis rotation in $A'$. This does not affect the entanglement based quantity $I^{(n)}$ that we need. In Step 3, a possible strategy for solving all $d_j$ is as follows: First obtain $I^{(2)}$ which gives the total number of anyon sectors $t$. Then obtain $I^{(n)}$ for more than $t$ number of additional Renyi indices, from which we can uniquely determine all $d_j/\mc D$. Since we know the smallest quantum dimension is that of the vacuum sector, $d_0=1$, we can subsequently find the values of $\mc D$ and all $d_j$. Note that for an abelian TO where $d_j=1$ for all $j$, $I^{(n)}=2\log\mc D$ is $n$-independent. Hence, if we are accessible to only a limited number of Renyi indices, although we are not able to obtain all $d_j$, we can still tell whether the TO is abelian or nonabelian. We shall also remark that Renyi EEs for different Renyi indices $n$ have rather different quantum information properties. For example, the strong subadditivity condition holds only for $n=1$ \cite{Berta2015RenyiCMI}. Our proof of Eq.\,\ref{eq:ProtocolFormula}, as we will see later, does not utilize such kind of quantum information property and thus holds for all $n$. 

For integer values of $n$, the quantity $I^{(n)}(11',33')$ proposed above can in principle be experimentally measured. To see this, we need to first understand how to prepare the state $\ket{\rho_A}$ in practice. 
Let $\{\ket{i} \}$ and $\{ \ket{\mu} \}$ be orthonormal bases of $A$ and $B$, respectively. We can write $\ket{\psi}=\sum_{i,\mu}\psi_{i\mu}\ket{i,\mu}$. It follows that $\ket{\rho_A}\propto\sum_{i,j,\mu}\psi_{i\mu}\psi^*_{j\mu}\ket{i}\ket{j}$. Denote by $\ket{\psi^*}=\sum_{i,\mu}\psi_{i\mu}^*\ket{i,\mu}$ the time-reversed copy of $\ket{\psi}$, i.e. the conjugate state in the chosen basis. We observe that $\ket{\rho_A}$ is proportional to $\bra{\Psi}(\ket{\psi}\otimes \ket{\psi^*})$, where $\ket{\Psi}\propto\sum_\lambda\ket{\lambda}\ket{\lambda}$ is a maximally entangled state living in two copies of $B$. This is illustrated using tensor diagrams in Fig.\,\ref{fig:TensorRep}. It means that to prepare $\ket{\rho_A}$, we may first prepare the state $\ket{\psi}\otimes \ket{\psi^*}$, and then implement a partial projection onto the state $\ket{\Psi}$ using projective measurements with postselections. If $\ket{\psi}$ can be prepared in a quantum simulation platform using unitary circuits and measurements, then it should be equally easy to prepare the time-reversed copy $\ket{\psi^*}$. 
Once $\ket{\rho_A}$ can be prepared, one can measure EEs (and therefore the mutual information) for integer Renyi indices $n\geq 2$ using established methods. For example, to measure the $n$-th Renyi EE of a subregion $R$ in a pure state, it suffices to measure the expectation value of the ``shift operator'' $C_n$. By definition, $C_n$ acts on $n$ copies of the same pure state, and its effect is to cyclically permute the $n$ copies of subregion $R$. We note that both the postselections required for obtaining $\ket{\rho_A}$ and the measurement of EEs require resources that scale exponentially with the system size. Nonetheless, this is not a problem in principle. Since we are dealing with gapped quantum systems with finite correlation lengths, there is no need to go to very large system sizes to get accurate results -- We just need the size of each subregion to well exceed the correlation length.

\begin{figure}
	\centering
	\includegraphics{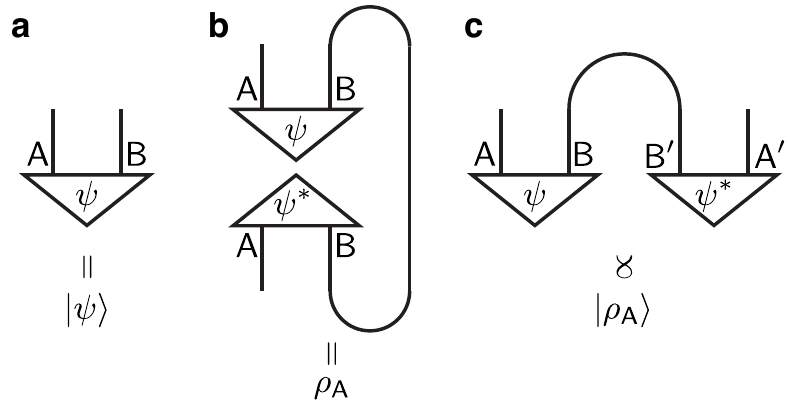}
	\caption{Tensor diagrams for the states used in the protocol. Panels a, b, and c represent $\ket{\psi}$, $\rho_A$, and $\ket{\rho_A}$, respectively. Panel c also illustrates a practical way of preparing the state $\ket{\rho_A}$. }
	\label{fig:TensorRep}
\end{figure}


\subsection*{Continuum Approach}
In this section, we will give a field-theoretic proof of Eq.\,\ref{eq:ProtocolFormula}, assuming the underlying TO to be described by a Chern-Simons (CS) theory \cite{Witten1989CSTheory}. The readers need not be familiar with CS theories, and just need to know that (1) a CS theory is a gauge theory with some compact gauge group $G$, and (2) it is a topological field theory, meaning that the action has no dependence on the spacetime metric and only the spacetime topology matters. 
As mentioned previously, we require the state $\ket{\psi}$ to have no excitations (zero energy density) in a large enough region. We expect that the reduced density matrix of $\ket{\psi}$ on a disk deep inside this region has no dependence on the boundary condition or excitations far away \cite{Cui2020QuantumDouble,Wen2016NegBdry,Ludwig2012TOEESpectrum}. 
Hence, for simplicity, we assume $\ket{\psi}$ to be the unique ground state of the TO on a sphere. In the CS theory, up to a normalization factor, this state can be prepared by performing the path integral in the interior of the sphere, i.e. on a solid ball, as illustrated in Fig.\,\ref{fig:Surgery}a. 
\begin{figure}
	\centering
	\includegraphics{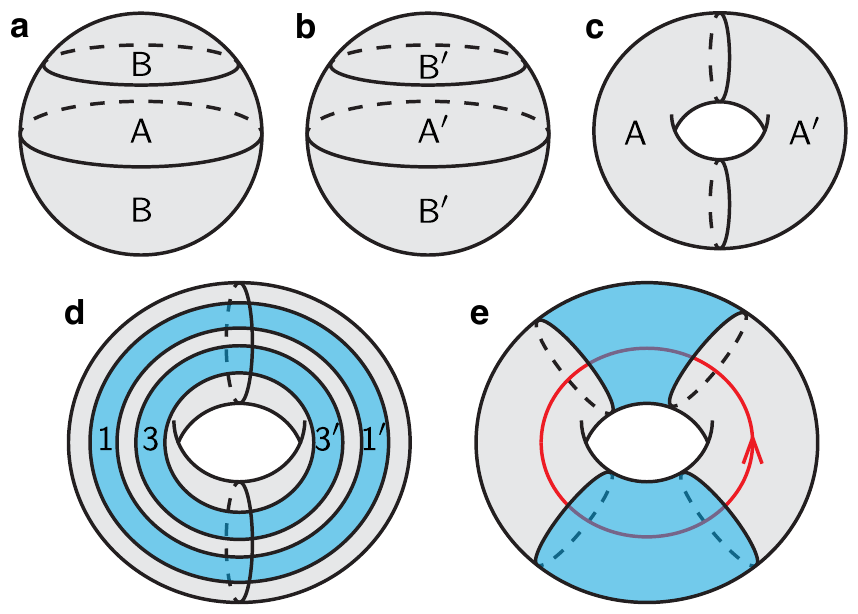}
	\caption{Illustration of the field-theoretic approach. (a) The sphere state $\ket{\psi}$ is prepared by doing the path integral on the solid ball as indicated by the grey shade. (b) Path integral for $\bra{\psi}$. (c) Path integral for $\rho_A$ or $\ket{\rho_A}$. (d) Regions $A_1\cup A_1'$ and $A_3\cup A_3'$. (e) The effect of an $\mc S$ transformation. }
	\label{fig:Surgery}
\end{figure}

Given the path integral representation of $\ket{\psi}$, we take a mirror image for $\bra{\psi}$ \footnote{One may check from the CS theory action that taking the complex conjugate of the wave function is equivalent to a mirror reflection of the spacetime manifold. } as shown in Fig.\,\ref{fig:Surgery}b, and then glue together $B$ of $\ket{\psi}$ and $B'$ of $\bra{\psi}$ to obtain the path integral for $\rho_A$. The result is shown in Fig.\,\ref{fig:Surgery}c. Up to a normalization, $\ket{\rho_A}$ has the same path integral representation as $\rho_A$, it is therefore a state living on the torus. 
The Hilbert space on a torus without anyon excitation is multidimensional. An orthonormal basis of the space, denoted by $\{ \ket{R_j} \}$, one-to-one corresponds to a finite set of representations $\{R_j\}$ of the gauge group, and also one-to-one corresponds to the anyon types of the TO. The state $\ket{R_j}$ can be prepared by performing the path integral on the solid torus (bagel) with a noncontractible Wilson loop \footnote{A Wilson loop is a certain observable defined on an oriented loop and labeled by a representation of the gauge group. It can be regarded as an anyon world line. } carrying the corresponding representation $R_j$ inserted. As shown in Fig.\,\ref{fig:Surgery}c, the path integral for $\ket{\rho_A}$ has no Wilson loop insertion. The state thus corresponds to the trivial representation, or the trivial anyon sector (vacuum). 

By keeping track of the subregions of $A$, we observe that $A_1\cup A_1'$ and $A_3\cup A_3'$ correspond to two annuli shown in Fig.\,\ref{fig:Surgery}d. We are now supposed to compute the Renyi mutual information between these two regions. To this end, it is convenient to first apply an $\mc S$ transformation \cite{Witten1989CSTheory}, whose effect is shown in Fig.\,\ref{fig:Surgery}e: The two annuli now wind in the perpendicular direction, and a Wilson loop is inserted in the path integral. This new torus state is given by $\sum_j\mc S_{0j}\ket{R_j}$, where $\mc S_{0j}=d_j/\mc D$ are components of the modular $\mc S$ matrix. The desired mutual information $I^{(n)}$ can now be computed using the replica trick and the surgery method \cite{Witten1989CSTheory,Fradkin2008TEE,Wen2016NegSurgery}. In fact, this has been done in Appendix B.4 of Ref.\,\onlinecite{Wen2016NegSurgery} (plug in $\psi_a=\mc S_{0a}$), and the technique is also pedagogically explained in that paper. We arrive at the result in Eq.\,\ref{eq:ProtocolFormula}. 

As a fixed-point theory, the CS theory only captures the universal terms in EEs. For a generic gapped field theory or lattice model, the EE of a region also contains nonuniversal terms such as the ``area-law'' term proportional to the length of region boundary, and terms due to corners or other sharp features which are inevitable on lattices. We need to discuss whether the quantity $I^{(n)}$ we consider contains any nonuniversal term. For a general gapped theory, we expect the picture of Fig.\,\ref{fig:Surgery}d still holds, although the theory is now not topological and depends on the spacetime metric. If we assume that nonuniversal terms in the EEs are made of {local} contributions (which are insensitive to changes far away) near the partition interfaces \cite{Kitaev2006TEE,Levin2006TEE}, then we see that all such terms have been canceled in $I^{(n)}$. 
For example, the boundary of $A_1\cup A_1'$ contributes the same nonuniversal terms to  $S^{(n)}_{11'}$ and $S^{(n)}_{11'\cup 33'}$, and these terms have been canceled in $I^{(n)}(11',33')$. We note that the locality assumption about nonuniversal terms does not hold in certain systems with the so-called {suprious} long-range entanglement \cite{Bravyi2008Spurious,Zou2016Spurious,Cano2015Spurious,Cheng2019Spurious,Schuch2019Spurious,Kato2020Spurious,Kim2023EELowerBound}, which will not be considered in this work. 
As one test of the universality of $I^{(n)}$, one can manually add a local bunch of coupled qubits to the state $\ket{\psi}$ at an arbitrary location, and observe that the final result of $I^{(n)}$ has no dependence on the state of these qubits. 

\subsection*{Test on Lattices}
In addition to the continuum approach, we have also tested our protocol on lattices using Kitaev's quantum double models \cite{Kitaev2003ToricCode}. This calculation is rather involved, so in the main text, we will only consider the simplest example -- that of the toric code model. The most general cases will be discussed in the Supplementary Information. 
\begin{figure}
	\centering
	\includegraphics{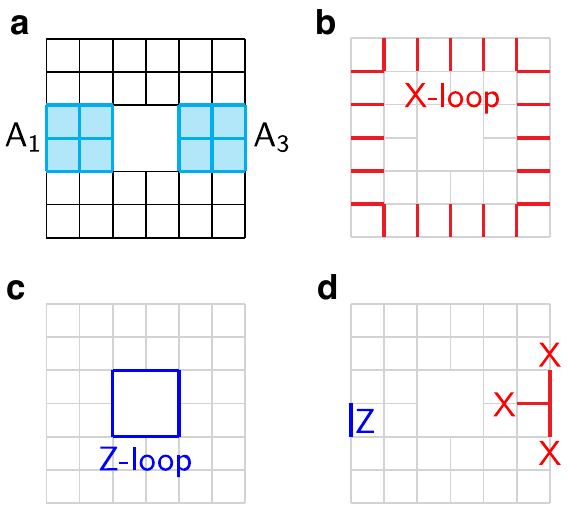}
	\caption{The toric code example. (a) The annulus region $A$. Subregions $A_1$ and $A_3$ (b) An $X$-loop operator. (c) A $Z$-loop operator. (d) Examples of boundary operators. }
	\label{fig:TCLattice}
\end{figure}

Given a square lattice with qubits living on the edges (links), the toric code model is defined by the following Hamiltonian.  
\begin{align}
H_{\rm TC}&=-\sum_v\dia{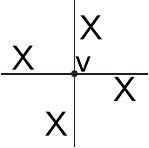}{-19}{0}-\sum_f\dia{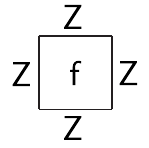}{-19}{0}\nonumber\\
&\equiv -\sum_v\mc X_v-\sum_f\mc Z_f. 
\label{eq:TCHamiltonian} 
\end{align}
The two set of terms in $H_{\rm TC}$ are usually called star and plaquette terms, respectively. Each star term $\mc X_v$ (plaquette term $\mc Z_f$) is the product of all Pauli-$X$ (Pauli-$Z$) operators surrounding a vertex $v$ (face $f$). These terms all commute with each other, and a ground state of $H_{\rm TC}$ is a simultaneous eigenstate of them with eigenvalue $+1$. 
It is not hard to generalize this definition to more general lattices, hence we can put the toric code model on surfaces of different topologies. On a sphere, $H_{\rm TC}$ has a unique ground states, but on a topologically nontrivial space such as a torus, $H_{\rm TC}$ has degenerate ground states. 

For simplicity, let us try implementing our protocol on the unique ground state $\ket{\Omega}$ on a sphere. For later convenience, we denote by $G$ the abelian group generated by all $\mc X_v$ and $\mc Z_f$ operators. $G$ is an example of the so-called stabilizer groups, and elements of $G$ are called stabilizers \cite{NielsenChuangBook}. Let the total number of qubits be $N$. $G$ can be generated by $N$ number of independent stabilizers $\{s_1,s_2,\cdots,s_N\}$; for example, we can take this set to be all but one $\mc X_v$ together with all but one $\mc Z_f$ (since $\prod_v\mc X_v=\prod_f\mc Z_f=1$). 
The full density matrix $\ket{\Omega}\bra{\Omega}$ can be interpreted as the projector onto the one-dimensional eigensubspace of $+1$ eigenvalue for all stabilizers in $G$. Hence, 
\begin{align}
	\ket{\Omega}\bra{\Omega}=\prod_{i=1}^N\left(\frac{1+s_i}{2}\right)=\frac{1}{2^N}\sum_{g\in G}g. 
\end{align}

We take an annulus region $A$ as shown in Fig.\,\ref{fig:TCLattice}a. Let $G_A\subset G$ be the subset of stabilizers acting in $A$, we have
\begin{align}
	\rho_A=\frac{1}{2^{N_A}}\sum_{g\in G_A}g,  
\end{align}
where $N_A$ is the number of qubits in $A$. $G_A$ is generated by all star and plaquette terms in $A$ as well as two loop operators shown in Fig.\,\ref{fig:TCLattice}, panels (b) and (c). We denote the two loop operators as $\mc W_X$ and $\mc W_Z$, respectively. 
The $Z$-loop ($X$-loop) operator $\mc W_Z$ ($\mc W_X$) is the product of all Pauli-$Z$ (Pauli-$X$) operators along a loop living on edges (dual lattice edges). 

We claim that the state $\ket{\rho_A}$ is a ground state of $H_{\rm TC}$ on a torus, where the torus is obtained by taking two copies of $A$ and identifying their corresponding boundary vertices. To prove this claim, we need to verify that $\ket{\rho_A}$ has $+1$ eigenvalue under all star and plaquette terms on the torus. We observe that all these terms have either of the following two forms, where we use $\otimes$ to connect operators acting on the two copies of $A$. 
\begin{itemize}
	\item $\mc O\otimes 1$ or $1\otimes \mc O$, where $\mc O$ is an $\mc X_v$ or $\mc Z_f$ operator acting in $A$. 
	\item $\Delta\otimes\Delta$, where $\Delta$ acts near the boundary of $A$, and two examples of $\Delta$ are given in Fig.\,\ref{fig:TCLattice}b. 
\end{itemize}
$\ket{\rho_A}$ satisfies the first set of stabilizers since $\mc O\rho_A=\rho_A\mc O=\rho_A$, where we used the fact that $\mc X_v$ and $\mc Z_f$ are both real Hermitian operators. $\ket{\rho_A}$ satisfies the second set of stabilizers because all boundary operators $\Delta$ commute with $G_A$ and thus $\Delta\rho_A\Delta=\rho_A$. This finishes the proof of the claim. On a torus, $H_{\rm TC}$ has four degenerate ground states. $\ket{\rho_A}$ can be uniquely determined by specifying two more stabilizers such as $\mc W_X\otimes 1$ and $\mc W_Z\otimes 1$. 

It remains to compute the Renyi mutual information $I^{(n)}(A_1\cup A_1',A_3\cup A_3')$, where $A_1$ and $A_3$ are shown in Fig.\,\ref{fig:TCLattice}a. Renyi EE and therefore mutual information can be conveniently computed in the stabilizer formalism \cite{Chuang2004StabilizerEE}. Let $\ket{\psi}$ be an $M$-qubit stabilizer state determined by a stabilizer group $H$. Let $R$ be a subregion with $M_R$ number of qubits and $H_R\subset H$ be the subgroup of stabilizers in $R$. From the reduced density matrix $\rho_R:=\Tr_{\bar R}(\ket{\psi}\bra{\psi})=2^{-M_R}\sum_{h\in H_R}h$, one can check that $S^{(n)}_R=(M_R-k_R)\log2$ where $k_R$ is the number of independent generators of $H_R$, i.e. $|H_R|=2^{k_R}$. The mutual information between two regions $R_1$ and $R_2$ is therefore given by $(k_{R_1\cup R_2}-k_{R_1}-k_{R_2})\log2$, independent on the Renyi index $n$. In our case, $R_1=A_1\cup A_1'$ and $R_2=A_3\cup A_3'$ are two annuli on the torus. $H_{R_1\cup R_2}$ has two more generators than $H_{R_1}H_{R_2}$: We can take the first (second) generator as the product of two noncontractible $X$-loop ($Z$-loop) operators in $R_1$ and $R_2$, respectively. 
We thus find $I^{(n)}(R_1,R_2)=2\log2$ for all $n$. This is indeed consistent with the fact that toric code is an abelian TO with $\mc D=2$. 

For general quantum double models, we find more interesting results of the mutual information, revealing nontrivial quantum dimensions. We refer interested readers to the Supplementary Information for details.

\section*{Discussion}
In this work, we have introduced a simple protocol for extracting all anyon quantum dimensions of a two-dimensional TO from an arbitrary ground state wave function. 
It is both validated in the continuum and verified on lattices. It is interesting to seek generalizations of this protocol for extracting more universal information, such as the fusion rules, $\mc S$ matrix, and topological spins. 

We should mention that this work is partially inspired by Ref.\,\onlinecite{Ryu2023TripartitionDiagrammatic}, which studies the entanglement negativity between two spatial regions in a tripartite topologically ordered state with trisection points (points where the three regions meet). Using some ``wormhole'' approach, it is found that the negativity contains an order-$1$ term that can distinguish abelian and nonabelian TOs. However, it is not clear at least to us whether this term is comparable to any universal quantity in generic models. 
It actually seems hard to extract a universal term from the entanglement negativity with trisection points \cite{LiuCC2022IQHNeg,Ryu2023Trisection,YinLiu2023Ent}, and more studies are needed to better understand this issue. 

Finally, we note that this work is still not totally satisfactory: Our protocol has only been checked using fixed-point models, either in the continuum or on lattices. Future numerical simulations are needed to further verify this protocol in the presence of perturbations. 

\section*{Data Availability}
This research is analytical; there is no numerical or experimental data. Part of the analytical derivations are provided in the Supplementary Information file. 

\section*{Code Availability}
Not applicable. 

\bibliography{Bib_Annulus.bib}

\section*{Acknowledgments}
I am grateful to Chao Yin for a previous related collaboration \cite{YinLiu2023Ent}, to Yanting Cheng and Pengfei Zhang for feedbacks on the manuscript, and to Wenjie Ji, Yuan-Ming Lu, Nat Tantivasadakarn, Ashvin Vishwanath, and Liujun Zou for helpful discussions. 
I am supported by the Gordon and Betty Moore Foundation under Grant No. GBMF8690 and the National Science Foundation under Grant No. NSF PHY-1748958. 

\section*{Author Contributions}
This is a single-author paper. 

\section*{Competing Interests}
The author declares no competing interests. 


\clearpage
\begin{center}
	\large\textsc{Supplementary Information}
\end{center}
\setcounter{equation}{0}
\setcounter{figure}{0}
\setcounter{table}{0}
\setcounter{page}{1}
\makeatletter
\renewcommand{\theequation}{S\arabic{equation}}
\renewcommand{\thefigure}{S\arabic{figure}}
\renewcommand{\bibnumfmt}[1]{[S#1]}

\appendix

In this appendix, we apply our protocol to Kitaev's quantum double models \cite{Kitaev2003ToricCode}. We first review the models in Supplementary Note \ref{App:QDModelReview} and then implement our protocol in Supplementary Note \ref{App:ProtocolonQDs}. In the last step of the protocol, we need some EE results in quantum double models. Details of the EE calculations are provided in Supplementary Note \ref{App:EEQDs}. 

\section{Quantum Double Models}\label{App:QDModelReview}
Let us start by reviewing some definitions. There is a quantum double model for each finite group $G$ (generally nonabelian). The model can live on an arbitrary lattice on an arbitrary orientable two-dimensional surface. The physical degrees of freedom, called spins, live on the edges, and the local Hilbert space of each spin is spanned by the orthonormal group element basis $\{ \ket{g}|g\in G \}$. We need to choose a direction for each edge. Reversing the direction of a particular edge will be equivalent to the basis change $\ket{z}\mapsto\ket{z^{-1}}$ for the corresponding spin. Let $v$ be a vertex, and $f$ be an adjacent face, we define the local gauge transformations $A_v(g)$ and magnetic charge operators $B_{(v,f)}(h)$ as follows. 
\begin{align}
	&A_v(g)\Bigg|{\dia{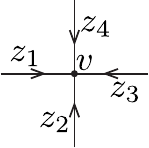}{-19}{0}}\Bigg\rangle=\Bigg|{\dia{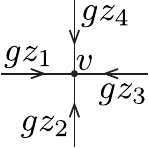}{-19}{0}}\Bigg\rangle,\\
	&B_{(v,f)}(h)\Bigg|{\dia{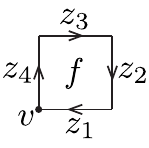}{-19}{0}}\Bigg\rangle=\delta_{z_1z_2z_3z_4,h}\Bigg|{\dia{FaceState}{-19}{0}}\Bigg\rangle. 
\end{align}
Here we use a tetravalent vertex and a square face as examples, and the generalizations should be straightforward. 
We further define two projectors:
\begin{align}
	A_v:=|G|^{-1}\sum_{g\in G}A_v(g),\quad  B_f:=B_{(v,f)}(1). 
\end{align}
Note that $B_f$ does not depend on the choice of the adjacent vertex $v$. The quantum double Hamiltonian is then given by \cite{Kitaev2003ToricCode}
\begin{align}
	H_{\rm QD}=-\sum_v A_v-\sum_f B_f. 
\end{align}
The projectors in $H_{\rm QD}$ all commute with each other. 

\begin{figure}
	\includegraphics{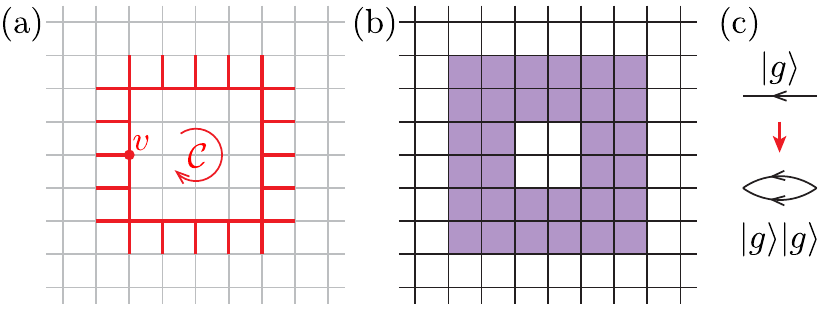}
	\caption{Illustrations about the quantum doubled models. (a) Support of the operator $A_{\mc C}(g)$ for a loop $\mc C$ with base point $v$. (b) Bipartition into $A$ (annulus) and $B$. (c) Duplication of a spin. }
	\label{fig:QDLattice}
\end{figure}

On a sphere, the model has a unique gapped ground state $\ket{\Omega}$ satisfying $A_v=B_f=1$ for all vertices $v$ and faces $f$. Let us introduce some useful properties of this state before implementing our protocol. More explicitly, we can write $\ket{\Omega}\propto(\prod_vA_v)(\bigotimes_e\ket{1}_e)$ where $e$ runs over all edges. Using $[A_v(g),A_{v'}(h)]=0$ for $v\neq v'$, and $A_v(g)A_v(h)=A_v(gh)$, one can check that 
\begin{itemize}
	\item $A_v(g)\ket{\Omega}=\ket{\Omega}$ for all $v$ and $g$. 
\end{itemize}
Given an oriented path along the edges from one vertex to another, when the relevant edge orientations are all consistent with the path direction, we define the {holonomy} of the path to be the product of all group elements on the path {in the reversed order}. For example, the holonomy of the following path from vertex $u$ to $v$ is given by $g_3g_2g_1$ (note the ordering). 
\[\dia{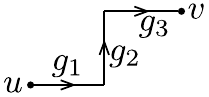}{-10}{0}\]
With this terminology, $B_f=1$ means that the holonomy around any face is trivial. It follows that 
\begin{itemize}
	\item for the state $\ket{\Omega}$, the holonomy for any closed loop is trivial, and the holonomy between any two vertices does not depend on the choice of path. 
\end{itemize}
Analogously, the property $A_v(g)\ket{\Omega}=\ket{\Omega}$ also has a generalization on loops. Given a loop $\mc C$ with a base point $v$, we can define operators $A_{\mc C}(g)$ whose support has the shape of a comb shown in Supplementary Fig.\,\ref{fig:QDLattice}a. The action of $A_{\mc C}(g)$ is defined by 
\begin{align}
	&A_{\mc C}(g)\bigg|\dia{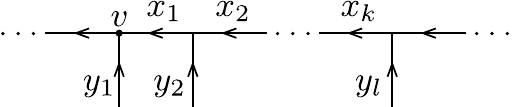}{-10}{0}\bigg\rangle\nonumber\\
	=&\bigg|\dia{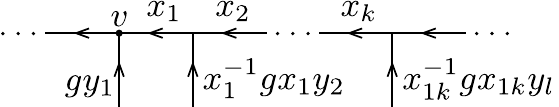}{-10}{0}\bigg\rangle, 
\end{align}
where $x_{1k}:=x_1x_2\cdots x_k$. Let $\mc L_{\rm flat}$ be the Hilbert subspace spanned by spin configurations satisfying $B_f=1$ $\forall f$. $A_{\mc C}(g)$ preserves $\mc L_{\rm flat}$. Moreover, $[A_{\mc C}(g),A_{v'}(h)]=0$ for $v\neq v'$, and $A_{\mc C}(g)A_v(h)=A_v(h)A_{\mc C}(h^{-1}gh)$. Let $A_{\mc C}:=|G|^{-1}\sum_{g\in G}A_{\mc C}(g)$. One can check $A_{\mc C}\ket{\Omega}=\ket{\Omega}$ with these commutation relations and the observation that acting $A_{\mc C}(g)$ on $\bigotimes_e\ket{1}_e$ is equivalent to acting several $A_{v}(g)$ operators. Using $A_{\mc C}(g)A_{\mc C}(h)=A_{\mc C}(gh)$, it follows that 
\begin{itemize}
	\item $A_{\mc C}(g)\ket{\Omega}=\ket{\Omega}$. 
\end{itemize}
We will also call $A_{\mc C}(g)$ a gauge transformation operator. 

\section{Applying the Protocol}\label{App:ProtocolonQDs}
With all these prerequisites, we are now ready to implement our protocol. Suppose a state $\ket{\psi}$ has no excitation with respect to $H_{\rm QD}$ in a large contractible region, and suppose this region has the form of a square lattice. It has been shown in Ref.\,\onlinecite{Cui2020QuantumDouble} that the reduced density matrix of $\ket{\psi}$ deep inside this region has no dependence on the choice of $\ket{\psi}$. Hence, we will just take $\ket{\psi}$ to be the sphere ground state $\ket{\Omega}$. 
Consider a bipartition of the space like that in Supplementary Fig.\,\ref{fig:QDLattice}b. Here, in order to draw the partition interface right on the edges, we imagine duplicating each edge into two, and require them to always be in the same group element state; see Supplementary Fig.\,\ref{fig:QDLattice}c. This is just a simple trick inspired by Ref.\,\onlinecite{Levin2006TEE} for getting a nice partition; each pair of spins obtained this way can still be regarded as a single spin unless the partition is being considered \footnote{We note that after this edge duplication, the resulting state can be regarded as the ground state of the quantum double model on the extended lattice. The requirement that each pair of spins obtained from duplication have the same group element state is ensured by the holonomy free condition. Therefore, the process of edge duplication does not change the TO of the state. }. Still calling the state $\ket{\Omega}$, we can write
\begin{align}
	\ket{\Omega}=\sum_{\substack{\text{$g_A$ with trivial}\\ \text{holonomies}}}\ket{g_A}_A\ket{\phi(g_A)}_B, 
\end{align}
where $g_A$'s are spin configurations in the annulus region $A$, and $\ket{\phi(g_A)}_B$ are some set of states on $B$ that are not necessarily normalized or orthogonal. In the summation above, we require $g_A$ to have trivial holonomy around any loop in $A$, contractible or noncontractible. 

We claim that for any two product states of group elements $\ket{g_A}$ and $\ket{g_A'}$ with trivial holonomies in $A$, and with the same group elements on the boundary $\pa A$, there exists a gauge transformation acting on $A$ which transforms $\ket{g_A}$ to $\ket{g_A'}$. We can build such a transformation step by step: First choose the unique local gauge transformation $A_{v_0}(g_0)$ acting on the top left internal vertex such that $A_{v_0}(g_0)\ket{g_A}$ matches $\ket{g_A'}$ on the entire top left face. Then move rightward and choose the unique $A_{v_1}(g_1)$ acting on the next vertex $v_1$ such that $A_{v_1}(g_1)A_{v_0}(g_0)\ket{g_A}$ matches $\ket{g_A'}$ on the top left two faces. Continue this process until the top row of faces are all done, and start over from the left most vertex in the second row. When the face at the top left corner of the internal boundary of $A$ is encountered, in order to not alter the spin configuration on $\pa A$, we need to utilize a loop operator $A_{\mc C}(g)$ to fix that face, where $\mc C$ coincides with the internal boundary. This kind of loop operators are no longer needed subsequently, and eventually $g_A$ can be completely transformed into $g_A'$. Since the $A_v(g)$ and $A_{\mc C}(h)$ operators are all unitary and leave $\ket{\Omega}$ invariant, when $g_A$ and $g_A'$ are related by a gauge transformation, we have $\ket{\phi(g_A)}_B={}_A\braket{g_A|\Omega}={}_A\braket{g_A'|\Omega}=\ket{\phi(g_A')}_B$. This means that $\phi(g_A)$ actually only depends on the spin configuration on $\pa A$. We can then write
\begin{align}
	\ket{\Omega}=\sum_{\substack{\text{$g_{\pa A}$ with trivial}\\ \text{holonomies}}}\ket{\xi(g_{\pa A})}_A\ket{\phi(g_{\pa A})}_B, 
\end{align}
where $\ket{\xi(g_{\pa A})}$ is the sum of all holonomy free $\ket{g_A}$ such that $(g_A)|_{\pa A}=g_{\pa A}$. 

The $\ket{\phi(g_{\pa A})}_B$ states are orthogonal to each other, because subsystem $B$ contains a copy of the spin configuration on $\pa A$. We will now prove that they also have the same norm. Observe that any spin configuration $g_{\pa A}$ with trivial holonomies can be transformed into another $g'_{\pa A}$ using the local gauge transformations $A_v(g)$ overlapping with $\pa A$. Let $U$ be that total gauge transformation operator. We can write $U=V_A V_B$ where $V_A$ ($V_B$) is a unitary operator acting on $A$ ($B$). From the definition of $\ket{\xi(g_{\pa A})}$, we can see $V_A\ket{\xi(g_{\pa A})}=\ket{\xi(g_{\pa A}')}$. It then follows from $U\ket{\Omega}=\ket{\Omega}$ that $V_B\ket{\phi(g_{\pa A})}=\ket{\phi(g_{\pa A}')}$. Hence $\ket{\phi(g_{\pa A})}$ and $\ket{\phi(g_{\pa A}')}$ indeed have the same norm. 

The above analysis is inspired by Ref.\,\onlinecite{Cui2020QuantumDouble}. With these results established, we easily find
\begin{align}
	\ket{\rho_A}\propto\sum_{\substack{\text{$g_{\pa A}$ with trivial}\\ \text{holonomies}}}\ket{\xi(g_{\pa A})}\ket{\xi(g_{\pa A})}. 
\end{align}
This state lives on the doubled system $A\cup A'$, i.e. two copies of the annulus in Supplementary Fig.\,\ref{fig:QDLattice}b. Now imagine gluing the vertices in $\pa A$ with the corresponding vertices in $\pa A'$, obtaining a torus. One can check that $\ket{\rho_A}$ is a ground state of the quantum double model defined on this torus. In particular, it is invariant under the actions of (new) local gauge transformations crossing the gluing interface. The quantum double model has multiple ground states on the torus. $\ket{\rho_A}$ is the special one characterized by trivial holonomy around the hole existing in each of the original annuli, and by $A_{\mc C}(g)=1$ around the same hole. These actually imply trivial anyon flux through the hole, because any anyon flux would be detected by some loop operator that winds another anyon around it. We have thus successfully recovered the picture in Fig.\,3c of the main text. 

The next step is to obtain the desired mutual information. Although this has been done in the continuum, we did not find a lattice result meeting our need. We have thus performed an honest calculation, and it eventually works out magically. We refer interested readers to Supplementary Note \ref{App:EEQDs} for the rather tedious details. The key technical trick is to use the holonomy basis introduced in Ref.\,\onlinecite{Wan2019TEE}: In the subspace with $B_f=1$, labeling spin configurations by group elements on all edges contains a lot of redundancy, and we can instead label spin configurations in each regions using independent holonomy variables. With this way of labeling basis vectors, the remaining calculation is more or less just brute-force. 

We have found that the anyon sectors are labeled by a pair of variables $(C,\mu)$. Here, $C$ labels a conjugacy class of $G$. Let $h_C\in C$ be a representative that is arbitrary but fixed once for all. $\mu$ labels an irreducible representation of $Z_C:=\{ g\in G|gh_C=h_C g\}$, the centralizer of $h_C$. The quantum dimensions are given by $d_{(C,\mu)}=|C|d_\mu$ where $d_\mu$ is the dimension of the representation $\mu$. These are consistent with known results \cite{Kitaev2003ToricCode,Hu2013TwistedQD,Wu2018LevinWenExcitations}. 

\section{EE Computations}\label{App:EEQDs}
In this appendix, we elaborate details about EE computations in quantum double models. We will utilize the technique of holonomy basis introduced in Ref.\,\onlinecite{Wan2019TEE}.
\subsection{Holonomy Basis}
Throughout this section, we restrict our attention to the Hilbert subspace $\mc L_{\rm flat}$ spanned by spin configurations satisfying $B_f=1$ for all $f$. Within this subspace, it is convenient to label spin configurations by independent holonomies: We choose some base point $v_0$, a path from $v_0$ to every other vertex, and a closed loop based at $v_0$ for each noncontractible cycle of the space. Then each spin configuration is uniquely determined by the holonomies along these paths and loops; one can solve the group element state on each edge from those holonomies. 

As an example, let there be $V$ number of vertices $v_0,w_1,w_2,\cdots,w_{V-1}$, and just one noncontractible circle. Denote by $g_i$ the holonomy from $v_0$ to $w_i$ along the chosen path, and by $k$ the holonomy around the closed loop. We can write the basis vectors as $\{\ket{g_i;k}\}$. An operator $A_{w_i}(h)$ acts as 
\begin{align}
	A_{w_i}(h)\ket{g_j;k}=\ket{g_1,\cdots,hg_i,\cdots,g_{V-1};k},
\end{align}
and $A_{v_0}(h)$ acts as
\begin{align}
	A_{v_0}(h)\ket{g_j;k}=\ket{g_1h^{-1},\cdots,g_{V-1}h^{-1};hkh^{-1}}. 
\end{align}

Note that despite the terminology ``holonomy basis'', we are still using the natural basis of tensor products of group elements. We just adopt a more convenient labeling of the basis vectors. 

\subsection{Warm-Up: EE of a Disk on a Sphere}
As a warm-up, let us compute the EE for a disk region on a sphere. Denote this disk by $A$, and the complement (also a disk) by $B$. We define holonomy bases separately for the two regions, as shown in Supplementary Fig.\,\ref{fig:DiskHolonomyBasis}. The base point in $A$ ($B$) is $v_A$ ($v_B$). Let there be $L$ number of vertices $w_1,\cdots,w_L$ on the partition interface, we denote the holonomy from $v_A$ ($v_B$) to $w_i$ by $g_i$ ($h_i$). There are also holonomies from $v_A$ ($v_B$) to other internal vertices of $A$ ($B$), but it turns out that these internal holonomies do not contribute to the EE. Therefore, for simplicity, we just retain one such internal vertex $u_A$ ($u_B$) in $A$ ($B$) and denote the corresponding holonomy by $p$ ($q$). 

\begin{figure}
	\includegraphics{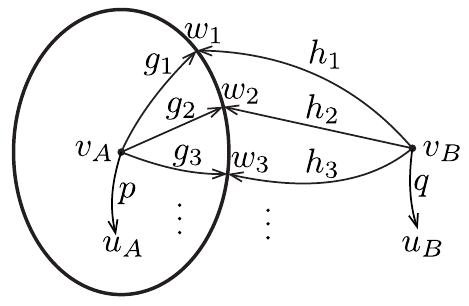}
	\caption{Holonomies for the inside and outside of a disk. }
	\label{fig:DiskHolonomyBasis}
\end{figure}

We are interested in states with $B_f=1$ for all faces $f$, which implies that the holonomy around any contractible loop is trivial. We thus have
\begin{align}
	g_1^{-1}h_1=g_2^{-1}h_2=\cdots=g_L^{-1}h_L=:a. 
\end{align}
Take an arbitrary holonomy configuration $\ket{g_i;h_i;p;q}$. Imposing the above condition, we can write $h_i=g_ia$. The ground state $\ket{\Omega}$ can be obtained by applying the projectors $A_v$ for all $v$ to the state $\ket{g_i;g_ia;p;q}$. First consider the $A_{w_i}$ operators. We have
\begin{align}
	&\prod_{i=1}^L A_{w_i}\ket{g_1,g_2,\cdots;g_1a,g_2a,\cdots;p;q}\\
	\propto & \sum_{h_i}\ket{h_1g_1,h_2g_2,\cdots;h_1g_1a,h_2g_2a,\cdots;p;q}\\
	=& \sum_{g_i}\ket{g_1,g_2,\cdots;g_1a,g_2a,\cdots;p;q}. 
\end{align}
Here, to obtain the last line, we first do a change of variable $h_i\mapsto h_ig_i^{-1}$ to absorb all $g_i$, and then rename the dummy variables $h_i$ to $g_i$. We see that the net effect of the $A_{w_i}$ operators is a summation over the $g_i$'s. Similarly, applying $A_{u_A}$ and $A_{u_B}$ results in a summation over $p$ and $q$. Next, applying $A_{v_B}$, we have
\begin{align}
	&A_{v_B}\sum_{g_i,p,q}\ket{g_i;g_ia;p;q}\\
	\propto&\sum_{g_i,p,q,h_B}\ket{g_i;g_iah_B^{-1};p;qh_B^{-1}}\\
	=&\sum_{g_i,p,q,a}\ket{g_i;g_ia;p;q}. 
\end{align}
We effectively get a summation over $a$. Finally, applying $A_{v_A}$ to the above state has no effect. We have thus found 
\begin{align}
	\ket{\Omega}\propto \sum_{g_i,p,q,a}\ket{g_i;g_ia;p;q}. 
\end{align}
We now see that the internal vertices $u_A$ and $u_B$ just contribute a product state $(|G|^{-1}\sum_p\ket{p})_A(|G|^{-1}\sum_q\ket{q})_B$ which does not affect the EE. We will therefore ignore the $p$ and $q$ variables in the following, and simply write $\ket{\Omega}\propto \sum_{g_i,a}\ket{g_i;g_ia}$. 

Taking a partial trace over subsystem $A$, we obtain
\begin{align}
	\rho_B\propto\sum_{g_i,a,b}\ket{g_ia}\bra{g_ib}. 
\end{align}
For later convenience, we do a change of variable: $a\mapsto g_1^{-1}a$, $b\mapsto g_1^{-1}b$, $g_{i>1}\mapsto g_{i>1}g_1$. It follows that 
\begin{align}
	\rho_B\propto \sum_{g_{i>1},a,b}\ket{a,g_2a,g_3a,\cdots}\bra{b,g_2b,g_3b,\cdots}.  
\end{align}
We denote by $\tilde\rho_B$ the unnormalized density matrix on the right-hand side. 

To compute the Renyi EE, we need to evaluate $\Tr(\tilde\rho_B^n)$. We add a superscript $\mu=1,2,\cdots,n$ to the dummy variales in the $\mu$-th copy of $\tilde\rho_B$, and it is not hard to see that 
\begin{align}
	b^\mu=a^{\mu+1},\quad g^\mu_{i>1}=g^{\mu+1}_{i>1}=:g_{i>1}. 
\end{align}
We are left with 
\begin{align}
	\Tr(\tilde\rho_B^n)=\sum_{a^\mu,g_{i>1}}1=|G|^{n+L-1}. 
\end{align}
The Renyi EE is then 
\begin{align}
	S_A^{(n)}=S_B^{(n)}=\frac{1}{1-n}\log\left[ \frac{\Tr(\tilde\rho_B^n)}{\Tr(\tilde\rho_B)^n} \right]=(L-1)\log|G|. 
\end{align}
We can identify the $L\log|G|$ term as the ``area-law'' term proportional to the length of the partition interface, and identify $-\log|G|$ as the universal topological EE \cite{Kitaev2006TEE,Levin2006TEE}. It follows that $\mc D=|G|$. 

Our calculation above essentially follows that in Ref.\,\onlinecite{Wan2019TEE}. The same result is also obtained in Ref.\,\onlinecite{Cui2020QuantumDouble} with different approaches.

\subsection{EEs on a Torus}
Let us now tackle the real problem. Recall that we need to compute the mutual information between two disjoint annulus regions of a torus. To the end, we need to compute the EEs both of a single annulus, and of two annuli. Let us then imagine partitioning the torus into $2N$ number of annuli, denoted by $T_1,T_2,\cdots,T_{2N}$. We will be interested only in the cases $N=1$ and $N=2$, but it will be convenient to keep a unified notation. 

We define holonomy bases separately for the $2N$ regions. In Supplementary Fig.\,\ref{fig:TorusHolonomyBasis}, we show our holonomy variables for the case of $N=2$, and generalizations to other values of $N$ should be clear. We denote by $v_I$ with $I\in\{1,2,\cdots,2N\}$ the base point in the region $T_I$. From each $v_I$, there are holonomies $g_{I,m}$ and $h_{I,m}$ to the interfaces with $T_{I-1}$ and $T_{I+1}$, respectively. We assume the numbers of vertices on all the interface circles are the same, denoted by $L$, just for the simplicity of notations. There is also a holonomy $k_I$ around a noncontractible loop in each region. As in the previous example, internal vertices will not contribute to the EE, so we have simply omitted them. We draw a loop $\mc C$ that is based at $v_1$ and passes through all the other $v_I$. The comb-like support of $ A_{\mc C}(g)$ operators associated with the loop is also plotted. As we have shown in the main text, the special ground state $\ket{\Omega}$ we consider has trivial holonomy around $\mc C$, and satisfies $A_{\mc C}\ket{\Omega}=\ket{\Omega}$. 

\begin{figure}
	\includegraphics{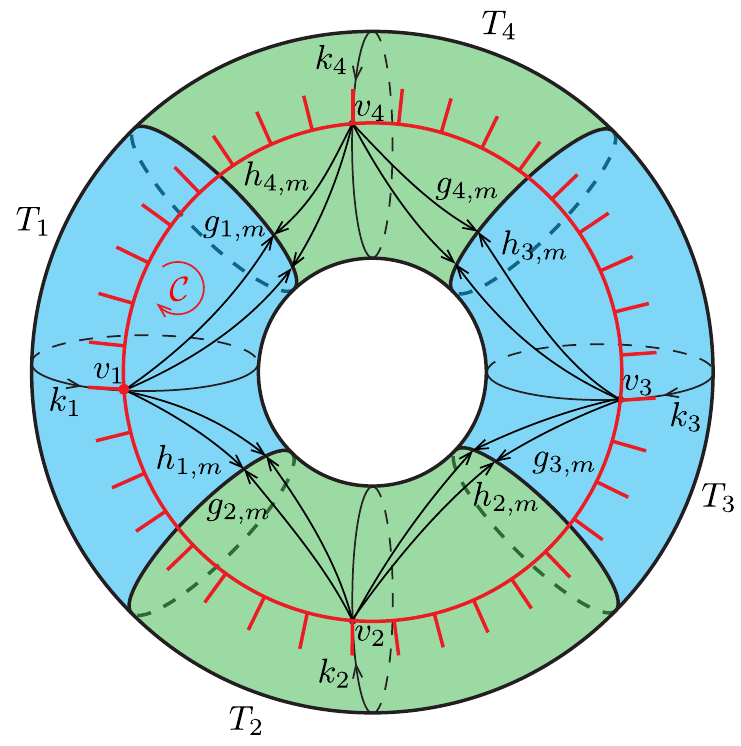}
	\caption{Holonomy variables for a tetrapartite torus, and the support (red comb) of the loop operators $A_{\mc C}(g)$. }
	\label{fig:TorusHolonomyBasis}
\end{figure}

Let us start from some holonomy configuration $\ket{g_{I,m};h_{I,m};k_I}$. We would like to impose $B_f=1$ for all faces $f$, or equivalently, every contractible loop has trivial holonomy. Then $g^{-1}_{I+1,m}h_{I,m}$ should be independent of $m$, and we will denote this quantity by $a_I$. In other words, $h_{I,m}=g_{I+1,m}a_I$. The $B_f=1$ conditions also imply $k_Ia_I^{-1}k_{I+1}^{-1}a_I=1$. We can thus write
\begin{align}
	k_{I+1}=a_I a_{I-1}\cdots a_1k_1a_1^{-1}\cdots a_{I-1}^{-1}a_I^{-1}. 
\end{align}
The condition of trivial holonomy around $\mc C$ implies that 
\begin{align}
	a_{2N}a_{2N-1}\cdots a_1=1. 
	\label{eq:TrivialCHolonomyApp}
\end{align}

To obtain the desired ground state $\ket{\Omega}$, we shall apply the projectors $A_{\mc C}$, $A_{I,m}$, and $A_{v_I}$ to $\ket{g_{I,m};h_{I,m};k_I}$, where $A_{I,m}$ is the local gauge invariance projector acting on the $m$-th vertex of the $T_IT_{I+1}$ interface (affecting $h_{I,m}$ and $g_{I+1,m}$). 

We can choose the paths for $g_{I,m}$ and $h_{I,m}$ to have no overlap with the support of $A_{\mc C}(z)$ operators. The same holds for paths from $v_I$ to other internal vertices, which we have omitted. As a result, $A_{\mc C}(z)$ will only affect $k_I$, and its action is more explicitly given by 
\begin{align}
	k_I\mapsto &(a_{I-1}\cdots a_1za_1^{-1}\cdots a_{I-1}^{-1})k_I\\
	&=a_{I-1}\cdots a_1(zk_1)a_1^{-1}\cdots a_{I-1}^{-1}. 
\end{align}
Hence, after the action of $A_{\mc C}$, the state $\ket{g_{I,m};h_{I,m};k_I}$ becomes
\begin{align}
	\sum_{k_1}\ket{g_{I,m};g_{I+1,m}a_I;a_{I-1}\cdots a_1k_1a_1^{-1}\cdots a_{I-1}^{-1}}, 
\end{align}
up to a normalization factor. The actions of $A_{I,m}$ effectively induce summations over $g_{I,m}$ for the above expression. Finally, applying $A_{v_I}$, we obtain
\begin{widetext}
	\begin{align}
		\ket{\Omega}\propto&\sum_{k_1,~g_{I,m},~h_I}\ket{g_{I,m}h_I^{-1};g_{I+1,m}a_Ih_I^{-1};h_Ia_{I-1}a_{I-2}\cdots a_1k_1a_1^{-1}\cdots a_{I-2}^{-1}a_{I-1}^{-1}h_I^{-1}}\\
		=&\sum_{k_1,~g_{I,m},~h_I}\ket{g_{I,m};g_{I+1,m}(h_{I+1}a_Ih_I^{-1});(h_Ia_{I-1}h_{I-1}^{-1})(h_{I-1}a_{I-2}h_{I-2}^{-1})\cdots (h_2a_1h_1^{-1})k_1\cdots}, 
	\end{align}
	where the second line is obtained by a change of variables: $g_{I,m}\mapsto g_{I,m}h_I$ and $k_1\mapsto h_1^{-1}k_1h_1$. Now observe that the following map within $G^{\times 2N}$ is one-to-one: 
	\begin{align}
		(h_{2N},h_{2N-1},\cdots,h_1)\mapsto(h_{2N}a_{2N-1}h_{2N-1}^{-1},h_{2N-1}a_{2N-2}h_{2N-2}^{-1},\cdots,h_2a_1h_1^{-1},h_1). 
	\end{align}
	Therefore, summing over the variables on the left is equivalent to summing over those on the right. Using this fact, together with Supplementary Eq.\,\ref{eq:TrivialCHolonomyApp} that implies $(h_1a_{2N}h_{2N}^{-1})(h_{2N}a_{2N-1}h_{2N-1}^{-1})\cdots(h_2a_1h_1^{-1})=1$, we obtain
	\begin{align}
		\ket{\Omega}\propto&\sum_{k_1,~a_I,~g_{I,m}}\delta(a_{2N}a_{2N-1}\cdots a_1,1)\ket{g_{I,m};g_{I+1,m}a_I;a_{I-1}\cdots a_1k_1a_1^{-1}\cdots a_{I-1}^{-1}}, 
	\end{align}
	where $\delta(\cdot,\cdot)$ is the Kronecker delta function. 
	
	Partial tracing $\ket{\Omega}\bra{\Omega}$ over $T_1,T_3,T_5,\cdots,T_{2N-1}$, and after a simple change of variables, we find the following reduced density matrix. 
	\begin{align}
		&\rho\propto\sum_{k,~g_{I,m},~a_I,~b_I}\delta(a_{2N}\cdots a_1,1)\delta(b_{2N}\cdots b_1,1)\delta(a_{2K-2}\cdots a_1 k a_1^{-1}\cdots a_{2K-2}^{-1},b_{2K-2}\cdots b_1 k b_1^{-1}\cdots b_{2K-2}^{-1})|_{2\leq K\leq N}\nonumber\\
		&\ket{g_{2K,m}a^{-1}_{2K-1};g_{2K+1,m}a_{2K};a_{2K-1}\cdots a_1 k a_1^{-1}\cdots a_{2K-1}^{-1}}\bra{g_{2K,m}b^{-1}_{2K-1};g_{2K+1,m}b_{2K};b_{2K-1}\cdots b_1 k b_1^{-1}\cdots b_{2K-1}^{-1}}=:\tilde\rho. 
	\end{align}
	We have introduced a new index $K$ which takes values from $\{ 1,2,\cdots,N \}$ unless otherwise indicated (like for the third $\delta$ above). Now apply the following change of variables: 
	\begin{align}
		&a_{2K-1}\mapsto a_{2K-1}g_{2K,1},\quad b_{2K-1}\mapsto b_{2K-1}g_{2K,1},\quad g_{2K,m>1}\mapsto g_{2K,m}g_{2K,1},\\
		&a_{2K}\mapsto g_{2K+1,1}^{-1}a_{2K},\quad b_{2K}\mapsto g_{2K+1,1}^{-1}b_{2K},\quad g_{2K+1,m>1}\mapsto g_{2K+1,m}g_{2K+1,1}.
	\end{align}
	We obtain
	\begin{align}
		\tilde\rho=&\sum_{k,~g_{I,m},~a_I,~b_I}\delta(g_{1,1}^{-1}a_{2N}a_{2N-1}g_{2N,1}g_{2N-1,1}^{-1}a_{2N-2}\cdots a_1g_{2,1},1)\delta(g_{1,1}^{-1}b_{2N}b_{2N-1}g_{2N,1}g_{2N-1,1}^{-1}b_{2N-2}\cdots b_1g_{2,1},1)\nonumber\\
		&\delta(g_{2K-1,1}^{-1}a_{2K-2}a_{2K-3}g_{2K-2,1}\cdots a_1g_{2,1}k\cdots,g_{2K-1,1}^{-1}b_{2K-2}b_{2K-3}g_{2K-2,1}\cdots b_1g_{2,1}k\cdots)|_{K>1}\nonumber\\
		&\ket{(1,g_{2K,m>1})a^{-1}_{2K-1};(1,g_{2K+1,m>1})a_{2K};a_{2K-1}g_{2K,1}g_{2K-1,1}^{-1}a_{2K-2}a_{2K-3}\cdots a_1g_{2,1} k \cdots}\nonumber\\
		&\bra{(1,g_{2K,m>1})b^{-1}_{2K-1};(1,g_{2K+1,m>1})b_{2K};b_{2K-1}g_{2K,1}g_{2K-1,1}^{-1}b_{2K-2}b_{2K-3}\cdots b_1g_{2,1} k \cdots}, 
	\end{align}
	an incredibly complicated expression! This step is similar to the one we took in the warm-up example before evaluating $\Tr(\tilde\rho_B^n)$. Here and below, we often write terms of the form $xkx^{-1}$ as $xk\cdots$ when $x$ has a long expression. 
	In the first entry of the ket, the notation $(1,g_{2K,m>1})$ represents a list of group elements indexed by $m$. $m=1$ corresponds to the identity $1$, and each $m>1$ corresponds to $g_{2K,m>1}$. These elements are then all multiplied by $a_{2K-1}^{-1}$ from the right. Other similar entries can be interpreted analogously. 
	Now we are ready to compute $\Tr(\tilde\rho^n)$. We again add a superscript $\mu$ to the dummy variables of the $\mu$-th copy of $\tilde\rho$. We see that $b_I^\mu=a_I^{\mu+1}$. The summations over $b_I^\mu$ can then be dropped. In addition, $g_{I,m>1}^\mu=g_{I,m>1}^{\mu+1}$. We can now do the summations over $g_{I,m>1}^\mu$ and get a factor $|G|^{2N(L-1)}$. We rename $g^\mu_{I,1}$ as $g^\mu_I$, and we are left with 
	\begin{align}
		\Tr(\tilde\rho^n)&=\sum_{k^\mu,g_{I}^{\mu},a^\mu_I}|G|^{2N(L-1)}\nonumber\\
		&\delta[(g_{1}^\mu)^{-1}a^\mu_{2N}a^\mu_{2N-1}g^\mu_{2N}(g^\mu_{2N-1})^{-1}a^\mu_{2N-2}\cdots a^\mu_1g^\mu_{2},1]
		\delta[(g_{1}^\mu)^{-1}a^{\mu+1}_{2N}a^{\mu+1}_{2N-1}g^\mu_{2N}(g^\mu_{2N-1})^{-1}a^{\mu+1}_{2N-2}\cdots a^{\mu+1}_1g^\mu_{2},1]\nonumber\\
		&\delta[(g^\mu_{2K-1})^{-1}a^\mu_{2K-2}a^\mu_{2K-3}g^\mu_{2K-2}\cdots a^\mu_1g^\mu_{2}k^\mu\cdots,(g^\mu_{2K-1})^{-1}a^{\mu+1}_{2K-2}a^{\mu+1}_{2K-3}g^\mu_{2K-2}\cdots a^{\mu+1}_1g^\mu_{2}k^\mu\cdots]|_{K>1}\nonumber\\
		&\delta[a^{\mu+1}_{2K-1}g^\mu_{2K}(g^\mu_{2K-1})^{-1}a^{\mu+1}_{2K-2}a^{\mu+1}_{2K-3}\cdots a^{\mu+1}_1g^\mu_{2} k^\mu \cdots,
		a^{\mu+1}_{2K-1}g^{\mu+1}_{2K}(g^{\mu+1}_{2K-1})^{-1}a^{\mu+1}_{2K-2}a^{\mu+1}_{2K-3}\cdots a^{\mu+1}_1g^{\mu+1}_{2} k^\mu \cdots]
		\label{eq:TraceTildeRhonGeneral}
	\end{align}
	It seems too hard to proceed with this expression for a general $N$, so in the following, we will restrict to $N=1$ and $N=2$ which are the cases we need. 
\end{widetext}
\subsubsection{{N}=1}
When $N=1$, $K$ can only be $1$, and we find
\begin{align}
	\Tr(\tilde\rho^n)&=\sum_{k^\mu,g^\mu_I,a^\mu_I}|G|^{2L-2}\nonumber\\
	&\delta[(g^\mu_1)^{-1}a^\mu_2a^\mu_1g^\mu_2,1]\delta[(g^\mu_1)^{-1}a^{\mu+1}_2a^{\mu+1}_1g^\mu_2,1]\nonumber\\
	&\delta( a^{\mu+1}_1g_2^\mu k^\mu\cdots,a_1^{\mu+1}g_2^{\mu+1}k^{\mu+1}\cdots ). 
\end{align}
Notice that in the last delta function, $a^{\mu+1}_1$ appears in both entries and thus can be removed. 
We can make a change of variables: $a^\mu_1\mapsto (a_2^\mu)^{-1} a_1^\mu$, and $k^\mu\mapsto (g_2^\mu)^{-1}k^\mu g_2^\mu$. The $a_2^\mu$ summation can now be done and gives $|G|^n$. The third $\delta$ function becomes $\delta(k^\mu,k^{\mu+1})$. The summation over $k^\mu$ therefore gives a factor of $|G|$. We are left with
\begin{align}
	&\Tr(\tilde\rho^n)=\sum_{g_I^\mu,a_1^\mu}|G|^{2L-1+n}
	\nonumber\\
	&\quad\quad\quad\delta[(g_1^\mu)^{-1}a_1^\mu g_2^\mu,1]\delta[(g_1^\mu)^{-1}a_1^{\mu+1} g_2^\mu,1]. 
\end{align}
We can do the summation over $a^\mu_1$, and get 
\begin{align}
	\Tr(\tilde\rho^n)&=\sum_{g_I^\mu}|G|^{2L-1+n}\delta[g_1^\mu(g_2^\mu)^{-1},g_1^{\mu+1}(g_2^{\mu+1})^{-1}]\\
	&=\sum_{g_I^\mu}|G|^{2L-1+n}\delta[g_1^\mu,g_1^{\mu+1}]=|G|^{2L+2n}. 
\end{align}
The Renyi EE can now be computed: 
\begin{align}
	S^{(n)}_{T_2}=S^{(n)}_{T_1}=\frac{1}{1-n}\log\left[ \frac{\Tr(\tilde\rho^n)}{(\Tr\tilde\rho)^n} \right]=2L\log|G|. 
\end{align}
Recall that we have assumed the two interface circles between $T_1$ and $T_2$ to have the same length $L$. In general the $2L$ factor above should be replaced by the total interface length. We see that the Renyi EE contains only an area-law term, so the topological EE in this case actually vanishes. This is consistent with the field-theoretic result. 

\subsubsection{{N}=2}
Plugging $N=2$ and $K=1,2$ into Supplementary Eq.\ref{eq:TraceTildeRhonGeneral}, we find
\begin{align}
	\Tr(\tilde\rho^n)&=\sum_{k^\mu,g^\mu_{I},a^\mu_I}|G|^{4(L-1)}\nonumber\\
	&\delta[(g_{1}^\mu)^{-1}a^\mu_{4}a^\mu_3g_4^{\mu}(g_3^\mu)^{-1}a_2^{\mu}a^\mu_1g^\mu_{2},1]\nonumber\\
	&\delta[(g_{1}^\mu)^{-1}a^{\mu+1}_{4}a^{\mu+1}_3\cdots a^{\mu+1}_1g^\mu_{2},1]\nonumber\\
	&\delta[(g_3^\mu)^{-1}a^\mu_{2}a^\mu_{1}g^\mu_{2}k^\mu\cdots,(g_3^\mu)^{-1}a^{\mu+1}_{2}a^{\mu+1}_{1}g^\mu_{2}k^\mu\cdots]\nonumber\\
	&\delta[a^{\mu+1}_{1}g^\mu_{2} k^\mu \cdots,
	a^{\mu+1}_{1}g^{\mu+1}_{2}k^\mu \cdots]\nonumber\\
	&\delta[a^{\mu+1}_{3}g^\mu_{4}(g^\mu_{3})^{-1}a^{\mu+1}_{2}a^{\mu+1}_{1}g^\mu_{2} k^\mu \cdots,\nonumber\\
	&~~a^{\mu+1}_{3}g^{\mu+1}_{4}(g^{\mu+1}_{3})^{-1}a^{\mu+1}_{2}a^{\mu+1}_{1}g^{\mu+1}_{2} k^\mu \cdots]. 
\end{align}
The last two delta functions here come from the last delta function of Supplementary Eq.\,\ref{eq:TraceTildeRhonGeneral} with $K=1,2$, respectively. Utilizing the first two delta functions, the last one can be simplified to
\begin{align}
	\delta[(a_4^{\mu+1})^{-1}g_1^\mu k^\mu\cdots,(a_4^{\mu+1})^{-1}g_1^{\mu+1} k^{\mu+1}\cdots]. 
\end{align}
Further removing some obvious redundancies in the delta functions, we are left with 
\begin{align}
	\Tr(\tilde\rho^n)&=\sum_{k^\mu,g^\mu_{I},a^\mu_I}|G|^{4(L-1)}\nonumber\\
	&\delta[(g_{1}^\mu)^{-1}a^\mu_{4}a^\mu_3g_4^{\mu}(g_3^\mu)^{-1}a_2^{\mu}a^\mu_1g^\mu_{2},1]\nonumber\\
	&\delta[(g_{1}^\mu)^{-1}a^{\mu+1}_{4}a^{\mu+1}_3g_4^{\mu}(g_3^\mu)^{-1}a_2^{\mu+1}a^{\mu+1}_1g^\mu_{2},1]\nonumber\\
	&\delta[a^\mu_{2}a^\mu_{1}g^\mu_{2}k^\mu\cdots,a^{\mu+1}_{2}a^{\mu+1}_{1}g^\mu_{2}k^\mu\cdots]\nonumber\\
	&\delta[g^\mu_{2} k^\mu (g^\mu_{2})^{-1},
	g^{\mu+1}_{2}k^{\mu+1} (g^{\mu+1}_{2})^{-1}]\nonumber\\
	&\delta[g_1^\mu k^\mu(g_1^\mu)^{-1},g_1^{\mu+1} k^{\mu+1}(g_1^{\mu+1})^{-1}]. 
\end{align}
Now we do a change of variables: $g_4^\mu\mapsto g_4^\mu g_3^\mu$, $a_3^\mu\mapsto (a_4^\mu)^{-1}a_3^\mu$, $a_1^\mu\mapsto (a_2^\mu)^{-1}a_1^\mu$. The summations over $g_3^\mu$, $a_4^\mu$ and $a_2^\mu$ can now be done and give $|G|^{3n}$. We also do the summation over $g_4^\mu$ which reduces the first two delta functions into a single one. We get 
\begin{align}
	\Tr(\tilde\rho^n)&=|G|^{4(L-1)+3n}\sum_{k^\mu,g_1^\mu,g_2^\mu,a_1^\mu,a_3^\mu}\nonumber\\
	&\delta[a_1^{\mu}g_2^\mu(g_1^\mu)^{-1}a_3^{\mu},a_1^{\mu+1}g_2^\mu(g_1^\mu)^{-1}a_3^{\mu+1}]\nonumber\\
	&\delta[a^\mu_{1}g^\mu_{2}k^\mu\cdots,a^{\mu+1}_{1}g^\mu_{2}k^\mu\cdots]\nonumber\\
	&\delta[g^\mu_{2} k^\mu (g^\mu_{2})^{-1},
	g^{\mu+1}_{2}k^{\mu+1} (g^{\mu+1}_{2})^{-1}]\nonumber\\
	&\delta[g_1^\mu k^\mu(g_1^\mu)^{-1},g_1^{\mu+1} k^{\mu+1}(g_1^{\mu+1})^{-1}]. 
\end{align}
One more change of variables: $k^\mu\mapsto (g_2^\mu)^{-1}k^\mu g_2^\mu$, and $g_1^\mu\mapsto g_1^\mu g_2^\mu$. The summation over $g_2^\mu$ can be done to give $|G|^n$. Renaming $g_1^\mu$ as $g^\mu$, we get
\begin{align}
	\Tr(\tilde\rho^n)&=|G|^{4(L-1)+4n}\sum_{k^\mu,g^\mu,a_1^\mu,a_3^\mu}\nonumber\\
	&\delta[a_1^{\mu}(g^\mu)^{-1}a_3^{\mu},a_1^{\mu+1}(g^\mu)^{-1}a_3^{\mu+1}]\nonumber\\
	&\delta[a^\mu_{1}k^\mu(a^\mu_{1})^{-1},a^{\mu+1}_{1}k^\mu(a^{\mu+1}_{1})^{-1}]\delta[k^\mu,
	k^{\mu+1}]\nonumber\\
	&\delta[g^\mu k^\mu(g^\mu)^{-1},g^{\mu+1} k^{\mu+1}(g^{\mu+1})^{-1}]\\
	&=|G|^{4(L-1)+4n}\sum_{k,g^\mu,a_1^\mu,a_3^\mu}\nonumber\\
	&\delta[a_1^{\mu}(g^\mu)^{-1}a_3^{\mu},a_1^{\mu+1}(g^\mu)^{-1}a_3^{\mu+1}]\nonumber\\
	&\delta[a^\mu_{1}k(a^\mu_{1})^{-1},a^{\mu+1}_{1}k(a^{\mu+1}_{1})^{-1}]\nonumber\\
	&\delta[g^\mu k(g^\mu)^{-1},g^{\mu+1} k(g^{\mu+1})^{-1}]. 
\end{align}
The last delta function is equivalent to  
\begin{align}
	&\delta[g^1 k (g^1)^{-1},g^\mu k(g^\mu)^{-1}]|_{\mu>1}\nonumber\\
	=~&\delta[k,(g^1)^{-1} g^\mu k(g^\mu)^{-1}g^1]|_{\mu>1}. 
\end{align}
We do a change of variables: $g^{\mu>1}\mapsto g^1g^{\mu>1}$, and $a_3^\mu\mapsto g^1 a_3^\mu$. The summation over $g^1$ can now be done and gives $|G|$. We get 
\begin{align}
	\Tr(\tilde\rho^n)&=|G|^{4L-3+4n}\sum_{k,g^{\mu>1},a_1^\mu,a_3^\mu}\nonumber\\
	&\delta(a_1^1a_3^1,a_1^2a_3^2)\delta[a_1^{\mu}(g^\mu)^{-1}a_3^{\mu},a_1^{\mu+1}(g^\mu)^{-1}a_3^{\mu+1}]|_{\mu>1}\nonumber\\
	&\delta[a^\mu_{1}k(a^\mu_{1})^{-1},a^{\mu+1}_{1}k(a^{\mu+1}_{1})^{-1}]\nonumber\\
	&\delta[k,g^\mu k(g^\mu)^{-1}]|_{\mu>1}. 
\end{align}
Notice that although the second last delta function applies to all $\mu$, only $n-1$ values of $\mu$ give independent constraints. Hence, we can restrict that delta function to $2\leq \mu\leq n$. We further rewrite
\begin{align}
	\Tr(\tilde\rho^n)&=|G|^{4L-3+4n}\sum_{k,g^{\mu>1},a_1^\mu,a_3^\mu}\nonumber\\
	&\delta[(a_1^2)^{-1}a_1^1 a_3^1 (a_3^2)^{-1},1]\nonumber\\
	&\delta[(g^\mu)^{-1},(a_1^{\mu})^{-1}a_1^{\mu+1}(g^\mu)^{-1}a_3^{\mu+1}(a_3^{\mu})^{-1}]|_{\mu>1}\nonumber\\
	&\delta[k,(a^\mu_{1})^{-1}a^{\mu+1}_{1}k(a^{\mu+1}_{1})^{-1}(a^\mu_{1})]|_{\mu>1}\nonumber\\
	&\delta[k,g^\mu k(g^\mu)^{-1}]|_{\mu>1}. 
\end{align}
Now define some variables: 
\begin{align}
	&(x^1,x^2,x^3,\cdots,x^n)=\nonumber\\
	&[a_1^2,(a_1^2)^{-1}a_1^3,(a_1^3)^{-1}a_1^4,\cdots,(a_1^n)^{-1}a_1^1],\\
	&(y^1,y^2,y^3,\cdots,y^n)=\nonumber\\
	&[a_3^2,a_3^3(a_3^2)^{-1},a_3^4(a_3^3)^{-1},\cdots,a_3^1(a_3^n)^{-1}]. 
\end{align}
The maps from $a_1^\mu,a_3^\mu$ to $x^\mu,y^\mu$ are one-to-one, so we can change the summation variables to $x^\mu$ and $y^\mu$. The summations over $x^1$ and $y^1$ can then be done, giving $|G|^2$. We get 
\begin{align}
	\Tr(\tilde\rho^n)&=|G|^{4L-1+4n}\sum_{k,g^{\mu>1}, x^{\mu>1},y^{\mu>1}}\nonumber\\
	&\delta(x^2x^3\cdots x^n y^n\cdots y^3 y^2,1)\nonumber\\
	&\delta[(g^\mu)^{-1},x^\mu(g^\mu)^{-1}y^\mu]|_{\mu>1}\nonumber\\
	&\delta[k,x^\mu k (x^\mu)^{-1}]|_{\mu>1}\delta[k,g^\mu k(g^\mu)^{-1}]|_{\mu>1}. 
\end{align}
Using the second delta function, we find $(y^\mu)^{-1}=g^\mu x^\mu(g^\mu)^{-1}$. We then have
\begin{align}
	&\Tr(\tilde\rho^n)=|G|^{4L-1+4n}\sum_{k,g^{\mu>1}, x^{\mu>1}}\nonumber\\
	&\delta[x^2 x^3 \cdots x^n,g^2 x^2(g^2)^{-1} g^3 x^3(g^3)^{-1}\cdots g^n x^n(g^n)^{-1}]\nonumber\\
	&\delta[k,x^\mu k (x^\mu)^{-1}]|_{\mu>1}\delta[k,g^\mu k(g^\mu)^{-1}]|_{\mu>1}. 
\end{align}
The last two delta functions imply that $x^\mu$ and $g^\mu$ both belong to the centralizer of $k$, namely $Z(k):=\{ g\in G|gk=kg \}$. We obtain the result
\begin{align}
	&\Tr(\tilde\rho^n)=|G|^{4L-1+4n}\sum_{k\in G}~~\sum_{g^{\mu>1}, x^{\mu>1}\in Z(k)}\nonumber\\
	&\delta[x^2 x^3 \cdots x^n,g^2 x^2(g^2)^{-1} g^3 x^3(g^3)^{-1}\cdots g^n x^n(g^n)^{-1}]. 
\end{align}
We can no longer simplify this expression just by changing variables, but it can be computed using nonabelian Fourier transforms. We summarize the key step as the following lemma. 
\begin{lemma}
	Let $H$ be a finite group, $z_i$ and $h_i$ with $i=1,2,3,\cdots,r$ be group elements. Then, 
	\begin{align}
		&\sum_{h_i, z_i\in H}\delta[z_1z_2\cdots z_r,(h_1z_1h_1^{-1})(h_2z_2h_2^{-1})\cdots(h_r z_r h_r^{-1})]\nonumber\\
		&=|H|^{2r-1}\sum_{\mu\in{\rm Irrep}(H)}d_\mu^{2-2r}, 
	\end{align}
	where ${\rm Irrep}(H)$ is the set of irreducible representations of $H$, and $d_\mu$ is the dimension of the representation $\mu$. 
\end{lemma}
\begin{proof}
	Denote by $L^2(H)$ the Hilbert space generated by the orthonormal basis $\{ \ket{h}|h\in H \}$. Another orthonormal basis of $L^2(H)$ consists of the following states. 
	\begin{align}
		\ket{\mu;a,b}=\sqrt{\frac{d_\mu}{|H|}}\sum_{h\in H}\bar D^\mu_{ab}(h)\ket{h}. 
	\end{align}
	Here, $\mu\in{\rm Irrep}(H)$, $a,b\in\{1,2,\cdots,d_\mu\}$, $D^\mu(h)$ is the matrix for $h$ in the representation $\mu$, and $\bar D^\mu(h)$ is the complex conjugate of $D^\mu(h)$. 
	The orthonormality of this new basis lies in Schur's orthogonality relation: 
	\begin{align}
		\sum_{h\in H}\bar D^\mu_{ab}(h)D^\nu_{cd}(h)=\frac{|H|}{d_\mu}\delta_{\mu\nu}\delta_{ac}\delta_{bd}. 
	\end{align}
	The inverse map is given by 
	\begin{align}
		\ket{h}=\sum_{\mu\in {\rm Irrep}(H)}\sum_{a,b=1}^{d_\mu}\sqrt{\frac{d_\mu}{|H|}}D^\mu_{ab}(h)\ket{\mu;a,b}. 
	\end{align}
	This is called the nonabelian Fourier transform. 
	
	Denote by $\mc I$ the summation in the statement of the lemma. We can rewrite the delta function there as an inner product in $L^2(H)$: 
	\begin{align}
		\mc I=\sum_{h_i, z_i\in H}\braket{(h_1z_1h_1^{-1})\cdots(h_r z_r h_r^{-1})|z_1\cdots z_r}. 
	\end{align}
	Applying nonabelian Fourier transforms, 
	\begin{align}
		\mc I=\sum_{h_i,z_i}\sum_{\mu}\sum_{a,b=1}^{d_\mu}\frac{d_\mu}{|H|}\bar D^\mu_{ab}(h_1z_1h_1^{-1}\cdots)D^\mu_{ab}(z_1z_2\cdots z_r). 
	\end{align}
	It is now useful to introduce a diagrammatic language: 
	\begin{align}
		D^\mu_{ab}(h)=\dia{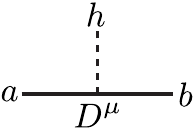}{-15}{0},\quad \bar D^\mu_{ab}(h)=\dia{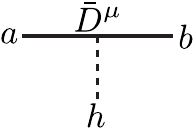}{-15}{0}. 
	\end{align}
	The orthogonality relation then looks like
	\begin{align}
		\dia{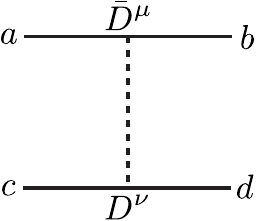}{-30}{0}=\frac{|H|}{d_\mu}\delta_{\mu\nu}\dia{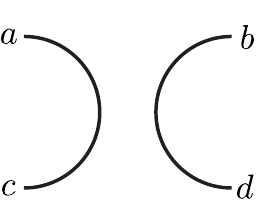}{-30}{0}. 
	\end{align}
	Using this diagrammatic language, and the fact that $D^\mu$ matrices satisfy the group multiplication laws, we compute the summations over $z_i$ as 
	\begin{align}
		&\sum_{z_i}\bar D^\mu_{ab}(h_1z_1h_1^{-1}\cdots)D^\mu_{ab}(z_1z_2\cdots z_r)\nonumber\\
		&=\dia{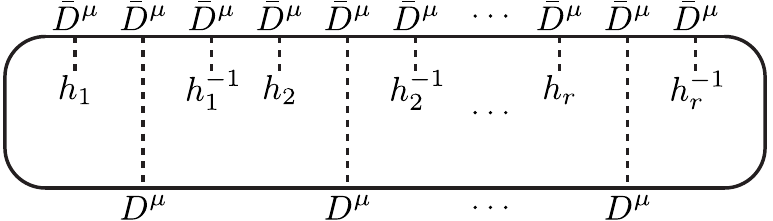}{-30}{0}\\
		&=\left(\frac{|H|}{d_\mu}\right)^r\dia{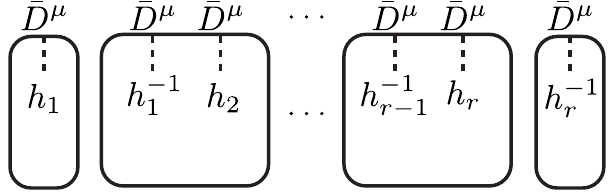}{-22}{0}. 
	\end{align}
	Since all the representations are unitary, we have $\bar D^\mu_{ab}(h^{-1})=D^\mu_{ba}(h)$. This implies, for example, 
	\begin{align}
		\dia{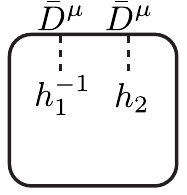}{-22.5}{0}=\dia{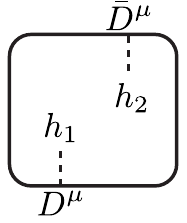}{-30}{0}. 
	\end{align}
	We can now do the summations over $h_i$, and find
	\begin{align}
		&\sum_{h_i}\sum_{z_i}\bar D^\mu_{ab}(h_1z_1h_1^{-1}\cdots)D^\mu_{ab}(z_1z_2\cdots z_r)\nonumber\\
		&=\left(\frac{|H|}{d_\mu}\right)^r\dia{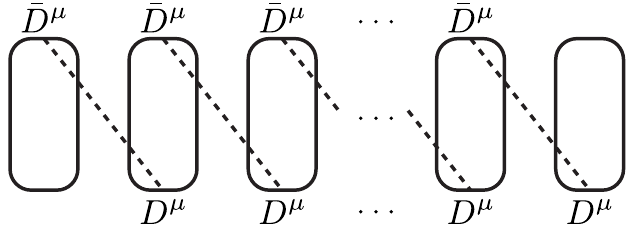}{-30}{0}\\
		&=\left(\frac{|H|}{d_\mu}\right)^{2r}\dia{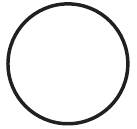}{-16}{0}=\left(\frac{|H|}{d_\mu}\right)^{2r}d_\mu
	\end{align}
	The lemma can be proved by plugging this result into the expression for $\mc I$. 
\end{proof}
We can now use this lemma to compute $\Tr(\tilde\rho^n)$, by identifying $r$ with $n-1$ and $H$ with $Z(k)$. Suppose $k$ belongs to the conjugacy class $C$. Then $Z(k)$ is isomorphic to $Z_C$ which is the centralizer of some arbitrary representative of $C$. We then find
\begin{align}
	\Tr(\tilde\rho^n)=|G|^{4L+6n-4}\sum_{(C,\mu)}\left(|C|d_\mu\right)^{4-2n},  
\end{align}
where $\mu\in {\rm Irrep}(Z_C)$, and we have used $|G|=|C||Z_C|$. In particular, $\Tr\tilde\rho=|G|^{4L+4}$, where we used $\sum_{\mu} d_\mu^2=|Z_C|$. We can then obtain the Renyi entropy
\begin{align}
	S^{(n)}_{T_2\cup T_4}&=S^{(n)}_{T_1\cup T_3}=\frac{1}{1-n}\log\left[ \frac{\Tr(\tilde\rho^n)}{(\Tr\tilde\rho)^n} \right]\\
	&=4L\log|G|+\frac{1}{1-n}\log\left[ \sum_{(C,\mu)}\left( \frac{d_{(C,\mu)}}{|G|}\right)^{4-2n} \right], 
\end{align}
where $d_{(C,\mu)}:=|C|d_\mu$. Recall again that we have assumed the four $T_iT_{i+1}$ interface circles to have the same length $L$. In general the $4L$ factor here should be replaced with the total interface length. 

\subsection{Mutual Information on a Torus}
Using the previous results, we get
\begin{align}
	I^{(n)}(T_1,T_3)=\frac{1}{n-1}\log\left[ \sum_{(C,\mu)}\left( \frac{d_{(C,\mu)}}{|G|}\right)^{4-2n} \right],  
\end{align}
suggesting that $d_{(C,\mu)}$ are nothing but the anyon quantum dimensions. 
\end{document}